\documentclass[twocolumn]{IEEEtran}
\usepackage{algorithmic,algorithm,flushend}
\usepackage{amsmath,amssymb,epsfig,color,flushend,algorithm}
\newtheorem{proposition}{Proposition}
\usepackage{times}

\usepackage{cite}

\usepackage{array}

\begin{document}
\title{Simultaneous Wireless Information and Power Transfer in MISO Full-Duplex Systems}
\author{Alexander A. Okandeji, Muhammad R. A. Khandaker, Kai-Kit Wong,\\Gan Zheng, Yangyang Zhang, and Zhongbin Zheng
\thanks{\scriptsize{A. A. Okandeji, M. R. A. Khandaker and K. K. Wong are with the Department of Electronic and Electrical Engineering, University College London, United Kingdom (e-mail: $\rm \{alexander.okandeji.13; m.khandaker; kai\text{-}kit.wong\}@ucl.ac.uk$). G. Zheng is with Wolfson School of Mechanical, Electrical and Manufacturing Engineering, Loughborough University, United Kingdom (e-mail: $\rm g.zheng@lboro.ac.uk$). Y. Zhang is with Kuang-Chi Institute of Advanced Technology, Shenzhen, China while Z. Zheng is with East China Institute of Telecommunications, China Academy of Information and Communications Technology, Shanghai, China.
\par This work was supported by the Presidential Special Scholarship Scheme for Innovation and Development (PRESSID), Federal Republic of Nigeria. This work was also supported in part by the EPSRC under grants EP/K015893/1 and EP/N008219/1.}}
}
\maketitle
\begin{abstract}
This paper investigates a multiuser multiple-input single-output (MISO) full-duplex (FD) system for simultaneous wireless information and power transfer (SWIPT), in which a multi-antenna base station (BS) simultaneously sends wirelessly information and power to a set of single-antenna mobile stations (MSs) using power splitters (PSs) in the downlink and receives information in the uplink in FD mode. In particular, we address the joint design of the receive PS ratio and the transmit power at the MSs, and the beamforming matrix at the BS under signal-to-interference-plus-noise ratio (SINR) and the harvested power constraints. Using semidefinite relaxation (SDR), we obtain the solution to the problem with imperfect channel state information (CSI) of the self-interfering channels. Furthermore, we propose another suboptimal zero-forcing (ZF) based solution by separating the optimization of the transmit beamforming vector and the PS ratio. Simulation results are provided to evaluate the performance of the proposed beamforming designs.
 \end{abstract}

 \section{Introduction}\label{sec_intro}
Since signals that carry information can also be used as a vehicle for transporting energy at the same time, simultaneous wireless information and power transfer (SWIPT) has attracted huge interest from industrial and academic communities \cite{ruhuul2,Alex,liao,ruhuul2a,Alex_3,num_5, weighted, power}. Through SWIPT, wireless data and energy access is made available to mobile users at the same time which brings great convenience. A practical application for SWIPT appears to be the wireless sensor networks mounted at remote and difficult-to-access locations and powered only by a battery with limited operation time \cite{mimo}. Recharging or replacing batteries often requires huge cost and can sometimes be inconvenient. Fittingly, the huge amount of electromagnetic energy in the environment resulting from numerous radio and television broadcasting may serve as a useful opportunistic as well as a greener alternative for harvesting energy to power such devices.

The fundamental concept of SWIPT was first introduced in \cite{varshney}. To describe the trade-off between the rates at which energy and dependable information is transmitted over a noisy channel, \cite{varshney} proposed a capacity-energy function. Work done in \cite{varshney} assumed that the receiver is capable of simultaneously decoding information and harvesting energy from the same received signal; however, this assumption does not hold in practice as practical circuits for harvesting energy from the received signal are not able to decode information directly. Consequently, to coordinate SWIPT at the receiver end, \cite{Architecture} developed a state-of-the-art receiver architecture and the corresponding rate-energy trade-off was analyzed.

To facilitate wireless information and power transfer at the receiver, \cite{mimo, Architecture} proposed two practical receiver architecture designs, namely time switching (TS) and power splitting (PS). In TS, the receiver switches over time to achieve information decoding and energy harvesting, while with PS, the receiver splits the received signal into two streams of different power in order to decode information and harvest energy separately. Theoretically, TS can be regarded as a special form of PS with only a binary split power ratios. Therefore, it suffices to say that PS generally achieves a better rate-energy transmission trade-off than TS \cite{mimo, Architecture}. However, in practical circuits, TS design architecture requires a simpler switch while PS requires a radio frequency (RF) signal splitter. Additionally, \cite{mimo} studied a scenario where a base station (BS) was broadcasting to two mobile users using TS. It is worth noting that the scenario in \cite{mimo} simplifies the receiver design but at a cost of compromising the efficiencies of perfect SWIPT technology. In contrast, the authors in \cite{Architecture} investigated SWIPT technology using practical receiver architecture for point-to-point systems.

A key concern for SWIPT is a decay of the power transfer efficiency which increases with the transmission distance. To overcome this challenge, the work in \cite{ruhuul2,mimo} exploited spatial diversity by multi-antenna techniques. Recently in \cite{ruhuul,QOS,joint_transmit}, SWIPT multicasting in multiple-input single-output (MISO), multiple-input multiple-output (MIMO) and SWIPT for MISO broadcasting systems were addressed. In particular, \cite{ruhuul,QOS} investigated the scenarios with perfect and imperfect channel state information (CSI) at the BS and presented algorithms for joint multicast transmit beamforming and receive PS for minimizing the BS transmit power subject to the quality-of-service (QoS) constraints at each receiver while \cite{joint_transmit} investigated the multi-antenna and PS enabled SWIPT system by considering a MISO broadcast channel consisting of one multi-antenna BS and a set of single antenna mobile stations (MSs).

Conventionally, wireless communication nodes operate in half-duplex (HD) mode in which transmission and reception of radio signals take place over orthogonal channels to avoid crosstalk. Recent advances, however, suggest that full-duplex (FD) communications that allows simultaneous transmission and reception of signals over the same frequency channel is possible \cite{exp, method_broad}. This brings a new opportunity for SWIPT in FD systems \cite{Alex, Alex2}. In addition to essentially doubling the bandwidth, FD communications also find additional source of energy in the inherent self-interfering signal.

However, a key challenge for realizing FD communication is to tackle self-interference (SI), which has to be significantly suppressed, if not cancelled completely, to the receiver's noise floor \cite{pfr}. Recently, much interest has been on the problem of SI cancellation in FD systems by investigating different system architectures and SI cancellation techniques to mitigate the self-interfering signal \cite{Transmit_strategies,practical,antenna,on_self,duarte}. The authors in \cite{Transmit_strategies} analyzed the transmit strategies for FD point-to-point systems with residual SI (RSI) and developed power adjustment schemes which maximize the sum-rate in different scenarios.

Digital SI cancellation (SIC) for FD wireless systems was studied in \cite{duarte}, in which it was demonstrated that for antenna separation and digital cancellation at 20cm and 40cm spacing between the interfering antennas, it was possible to achieve, respectively, 70dB and 76dB SIC. Motivated by the work in \cite{duarte}, \cite{exp} then showed that the effectiveness of SIC improves for active cancellation techniques as the received interference power increases. Furthermore, based on extensive experiments, it was revealed in \cite{exp} that a total average SIC of 74dB can be achieved. However, SIC cannot suppress the SI down to the noise floor \cite{exp,pfr}. In particular, \cite{pfr} proposed a digital SIC technique that can eliminate all transmitter impairments, and significantly mitigate the receiver phase noise and non-linearity effects. The technique was shown to mitigate the SI signal to $\approx$ 3dB above the receiver noise floor, which results to about $67\text{--}76\%$ rate improvement compared to conventional HD systems at 20dBm transmit power at the nodes. Accordingly, \cite{duarte} studied different mechanisms for SIC and showed that the power of the interfering signal can be estimated.

In the literature, SWIPT in FD systems have been considered in \cite{Alex,Alex2}. In \cite{Alex}, the problem of maximizing the sum-rate for SWIPT in FD bidirectional communication system was considered subject to energy harvesting and transmit power constraints at both nodes. On the other hand, \cite{Alex2} investigated SWIPT in FD MIMO relay system and addressed the optimal joint design of the receive PS ratio and beamforming matrix at the relay in order to maximize the achievable sum-rate.

In this paper, we extend the HD MISO SWIPT communication scenario in \cite{joint_transmit} to the FD case where the harvested energy at the MSs is utilized to send feedback information to the BS. We aim to minimize the end-to-end transmit power for SWIPT in FD MISO systems while satisfying the QoS requirements for each MS by optimizing jointly the receive PS ratio and the transmit power at the MSs, and the beamforming matrix at the BS. Specifically, for the MISO FD channel, we assume perfect CSI for the uplink and downlink channels as this is an idealization of actual practical systems. Effectively, perfect CSI can be accomplished from fine estimation, via the transmission of dedicated training symbols at the receiver. In contrast, we consider an imperfect CSI for loop channels due to the fact that the distribution of SI channels are unknown, and that the SI channel measurements results obtained in \cite{Ref_X} indicated that the SI channel has a multipath nature. These multiple paths can have higher power compared to the line of sight (LOS) path. This behaviour necessitates the need of an adaptive cancellation technique whose measurement is used to cancel both the LOS path and the delayed version of the same, which is not the aim of this work. However, it is a general practice to model the SI channels for simplicity as Gaussian channels \cite{Alex_3}. Accordingly, due to insufficient knowledge of the self-interfering channel, we consider a deterministic model for channel uncertainties where the magnitude of the estimation error as well as the SI power is bounded. Since the problem is non-convex, an alternating optimization approach is proposed. Semidefinite relaxation technique (SDR) is also adopted. Furthermore, we propose a zero-forcing (ZF) suboptimal solution by separating the optimization of the transmit beamforming vector at the BS, the PS ratios and the MS transmit power.

The rest of this paper is organized as follows. In Section \ref{sec_for}, the system model of a MISO SWIPT FD system with PS based MSs is introduced. The proposed joint optimization is devised in Section \ref{opt_soltn}. Section \ref{sub} introduces the ZF suboptimal solution while Section \ref{N_exam} provides the simulation results under various scenarios. Conclusions will be drawn in Section \ref{conc}.

{\em Notations}---We denote scalars by non-bold letters. Boldface lowercase letters are used to represent vectors, while boldface uppercase letters are for matrices. Further, ${\rm Tr}(\mathbf{A})$, ${\rm Rank}(\mathbf{A})$, $\mathbf{A}^T$ and $\mathbf{A}^H$ denote the trace, rank, transpose and conjugate transpose of ${\bf A}$, respectively, while $\mathbf{A}$ $\succeq 0$ means that $\mathbf{A}$ is a positive semidefinite matrix. Also, $\mathbf{I}_n$ denotes an $n \times n$ identity matrix, $\|\cdot€\|$ returns the Euclidean norm. The distribution of a circularly symmetric complex Gaussian (CSCG) vector with mean $\boldsymbol{\mu}$ and covariance matrix $\mathbf{C}$ is denoted by $\mathcal{CN}(\boldsymbol{\mu}, \mathbf{C}).$ Finally, $\mathbb{C}^{m \times n}$ is the space of $m \times n$ complex matrices.
\begin{figure}[ht]
\centering
\includegraphics*[width=8cm]{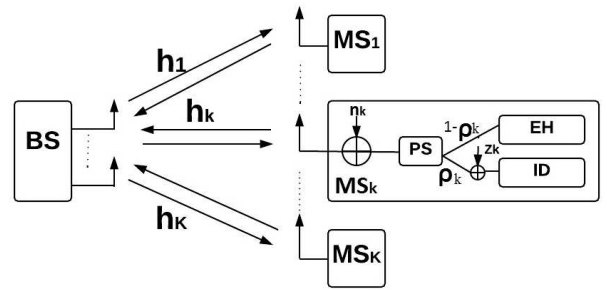}
\caption{A FD MISO SWIPT system.}\label{fig:sys}
\label{SI}
\end{figure}

\section{System Model and Problem Formulation}\label{sec_for}
In this paper, we study the end-to-end transmit power minimization approach for a multiuser MISO FD SWIPT system consisting of one BS and $K$ MSs, denoted by ${\rm MS}_1,\dots, {\rm MS}_K$, respectively, operating in FD mode as shown in Fig.~\ref{fig:sys}. The BS simultaneously transmits wireless information and power to a set of single antenna MSs in the downlink and receives information in the uplink in FD mode. We denote the number of transmit and receive antennas at the BS, respectively, as $N_t$ and $N_r$, and each MS uses an identical pair of transmitter and receiver antennas for signal transmission and reception. In the first phase, the BS performs transmit beamforming to send information to the MSs in the downlink while in the next phase, the MSs use the harvested energy from its own reception to send feedback information to the BS in the reverse link with a transmit power of $P_{{\rm up},k}$. The complex baseband transmitted signal at the BS can be expressed as
\begin{equation}
{\mathbf{x}}^{\mathrm{BS}} = \sum^K_{k = 1} {\mathbf{v}}_k s_k,
\end{equation}
where $s_k\sim\mathcal{CN} (0,1)$ denotes the transmitted information symbol to ${\rm MS}_k$, and $\mathbf{\mathbf{v}}_k$ represents the corresponding transmit beamforming vector. It is assumed that $s_k$, for $k = 1, \dots,€€€ K$, are independent and identically distributed (i.i.d.) CSCG random variables. We further assume quasi-static flat-fading for all MSs and denote $ {\mathbf{h}}_{{\rm dl},k}$ and ${{\mathbf{h}}}_{{\rm ul},k}$ as the conjugated complex channel vector from the BS to ${\rm MS}_k$ and from ${\rm MS}_k$ to the BS, respectively. The received signal at ${\rm MS}_k$ can be written as
\begin{equation}\label{SI^*}
y_k = \underbrace{{\mathbf{h}}^H_{{\rm dl},k} {\mathbf{v}}_k s_k}_{\mathrm{desired\hspace{0.7mm}signal}} + \underbrace{ \sum^K_{j \neq k}  {{\mathbf{h}}^H_{{\rm dl},k} {\mathbf{v}}_j s_j}}_{\mathrm{interfering\hspace{0.7mm}signal}} + \!\! \underbrace{{{h}}_{\mathrm{SI},k} m_k } _{\mathrm{self\text{-}interference}} \!\!\!\! +  n_k,
\end{equation}
where $m_k$ is the information carrying symbol of ${\rm MS}_k$ and $n_k\sim\mathcal{CN} (0, \sigma_k)$ denotes the antenna noise at the receiver of ${\rm MS}_k$. We assume that each MS is equipped with a PS device which coordinates the processes of information decoding and energy harvesting. In particular, we assume that the PS splits the received signal power such that a $\rho \in (0,1)$ portion of the signal power is fed to the information decoder (ID) and the remaining $(1-\rho)$ to the energy harvester (EH). Accordingly, the signal split to the ID of ${\rm MS}_k$ can be written as
\begin{multline}
y^{\mathrm{ID}}_k = \sqrt{\rho_k}
\left(\underbrace{{\mathbf{h}}^H_{{\rm dl},k} {\mathbf{v}}_k s_k}_{\mathrm{desired\hspace{0.7mm}signal}}\!\! +\!\! \underbrace{ \sum^K_{j \neq k}  {{\mathbf{h}}^H_{{\rm dl},k} {\mathbf{v}}_j s_j}}_{\mathrm{interfering\hspace{0.7mm}signal}} + \!\!\!\!\! \underbrace{ {{h}}_{\mathrm{SI},k} m_k } _{\mathrm{self\text{-}interference}}  \!\!\!\!\!\!\!+ n_k \right)\\
+ z_k,
\end{multline}
where $z_k \sim \mathcal{CN} (0,\delta^2_k)$ denotes the additional processing noise introduced by the ID at ${\rm MS}_k$. The signal split to the EH of ${\rm MS}_k$ is given by
\begin{equation}
y^{\rm EH}_k = \sqrt{1-\rho_k}\left( \sum^K_{j = 1}  {{\mathbf{h}}^H_{{\rm dl},j} {\mathbf{v}}_j s_j} + {{h}}_{\mathrm{SI},k} m_k     + n_k\right).
\end{equation}
Meanwhile, the signal received at the BS can be written as
\begin{equation}\label{SI2}
\mathbf{y}^{\mathrm{BS}} = \underbrace{ \sum^K_{k=1} {{\mathbf{h}}}_{{\rm ul},k} m_k }_{\mathrm{desired\hspace{0.7mm}signal}}  +  \underbrace{  \sum^K_{j = 1} { { {\mathbf{H}}}}_ {\mathrm{SI},\mathrm{BS}} {\mathbf{v}}_j s_j}_{\mathrm{self\text{-}interference}}  + {\mathbf{n}}_{\mathrm{BS}},
\end{equation}
where ${\mathbf{n}}_{\rm BS}\backsim \mathcal{CN} (\mathbf{0}, \sigma^2_\mathrm{{BS}} \mathbf{I})$ is the additive white Gaussian (AWGN) noise vector at the BS. To decode the signal from MS$_k,$ the BS applies a receive beamformer $\mathbf{w}_k$ to equalize the received signal from ${\rm MS}_k$ expressed as
\begin{multline}
s_k^{\rm UL} =\mathbf{w}^H_k \mathbf{h}_{{\rm ul},k} m_{k} +  \mathbf{w}^H_k \sum^K_{j \neq k} {{\mathbf{h}}}_{{\rm ul},_{j}} m_j \\
+{ \mathbf{w}^H_k  \sum^K_{j = 1}   { { {\mathbf{H}}}}_ {\mathrm{SI},\mathrm{BS}}  {\mathbf{v}}_j}   + \mathbf{w}^H_k {\mathbf{n}}_{\mathrm{BS}}.
 \end{multline}
The signal-to-interference plus noise ratio (SINR) at the BS from ${\rm MS}_k$ is therefore given by
\begin{multline}\label{SINR1}
\gamma_k^{\rm BS}=\\
\frac{  P_{up,k} | \mathbf{w}^H_k {{\mathbf{h}}}_{{\rm ul},k} \hspace{0.3mm}|^2}{  \sum^K_{j \neq k}\! P_{{\rm up},j} |\mathbf{w}^H_k \mathbf{h}_{{\rm ul},j}|^2 +\sum^K_{\!j = \! 1} |\mathbf{w}^{H}_k{{{\mathbf{H}}}}_{\mathrm{SI},\mathrm{BS}} {\mathbf{v}}_{j} |^2   +  \sigma_{\mathrm{BS}}^2\| \mathbf{w}_k\|^2}.
\end{multline}
Accordingly, the SINR at the ID of ${\rm MS}_k$ is given by
\begin{equation}\label{SINR2}
\gamma^{\rm MS}_k = \frac{\rho_k|{\mathbf{h}}^H_{{\rm dl},k} {\mathbf{v}}_k|^2}{\rho_k\left( \sum^K_{j \neq k} |{\mathbf{h}}^H_{{\rm dl},k} {\mathbf{v}}_j|^2  + |{{{h}}}_{\mathrm{SI},k}|^2 P_{{\rm up},k}  + \sigma^2_k \right) + \delta^2_k}.
\end{equation}
The harvested power by the EH of ${\rm MS}_k$ is given by
\begin{equation}\label{energy}
Q_k\! = \eta(1-\rho_k) \left( \sum^K_{j = 1}|{\mathbf{h}}^H_{{\rm dl},k} {\mathbf{v}}_j|^2 + |{{{h}}}_{\mathrm{SI},k} |^2 P_{{\rm up},k} + \sigma^2_k \right),
\end{equation}
where $\eta$ denotes the energy conversion efficiency at the EH of ${\rm MS}_k$ that accounts for the loss in energy transducer for converting the harvested energy into electrical energy to be stored. In practice, energy harvesting circuits are equipped at the energy harvesting receiver which are used to convert the received radio frequency (RF) power into direct current (DC) power. The efficiency of a diode-based EH is nonlinear and largely depends on the input power level \cite{valenta}. Hence, the conversion efficiency ($\eta$) should be included in the optimization expressions. However, for simplicity, we assume $\eta = 1$.

\subsection{Modelling SI}
Considering the fact that RSI cannot be eliminated completely due to the insufficient knowledge of the underlying channel, we consider a deterministic model for imperfect self-interfering channels. In particular, it is assumed that the SI channels ${{{h}}}_{\mathrm{SI},k}, \forall k,$ and ${{{\mathbf{H}}}}_{\mathrm{SI},\mathrm{BS}}$ lie in the neighbourhood of the estimated channels ${{\hat{h}}}_{\mathrm{SI},k}, \forall k,$ and ${{\hat{\mathbf{H}}}}_{\mathrm{SI},\mathrm{BS}},$ respectively, that are available at the nodes. Thus, the actual channels due to imperfect channel estimate can be modelled as
\begin{subequations}
\begin{align}
h_{\mathrm{SI},k} & = \hat{h}_{\mathrm{SI},k} + \triangle{h}_{\mathrm{SI},k},\label{SI_a} \\
\mathbf{H}_{\mathrm{SI},\mathrm{BS}}   &=  \hat{\mathbf{H}}_{\mathrm{SI},\mathrm{BS}} + \triangle \mathbf{{H}}_{\mathrm{SI},\mathrm{BS}},\label{SI_b}
\end{align}
\end{subequations}
where $\triangle{h}_{\mathrm{SI},k}$ and $\triangle \mathbf{H}_{\mathrm{SI},\mathrm{BS}}$ represent the channel uncertainties which are assumed to be bounded as
\begin{subequations}
\begin{align}
|\triangle{h}_{\mathrm{SI},k}| &= | {h}_{\mathrm{SI},k} - \hat{h}_{\mathrm{SI},k}| \leq \epsilon_1,\\
\|\triangle{\mathbf{H}}_{\mathrm{SI},\mathrm{BS}} \| & = \| {\mathbf{H}}_{\mathrm{SI},\mathrm{BS}} - \hat{ \mathbf{{H}}}_{\mathrm{SI},\mathrm{BS}}\| \leq \epsilon_2,
\end{align}
\end{subequations}
for some $\epsilon_1,\epsilon_2\geq 0$. The bounding values $\{\epsilon_k\}$ depend on the accuracy of the CSI estimates. To efficiently define the worst-case SI level, we modify (\ref{SI_a}) and (\ref{SI_b}) using the triangle inequality and the Cauchy-Schwarz inequality, respectively \cite{boyd}. It follows from (\ref{SI_a}) and (\ref{SI_b}) that
\begin{subequations}
\begin{eqnarray}
 |{h}_{\mathrm{SI},k}|^2 \! \!\!\!& = &\!\!\! | (\hat{h}_{\mathrm{SI},k} + \triangle{h}_{\mathrm{SI},k})|^2
  \leq (| \hat{h}_{\mathrm{SI},k}| + |\triangle{h}_{\mathrm{SI},k}|)^2 \nonumber\\
\!\!\!& \leq &\!\!\! |\hat{h}_{\mathrm{SI},k}|^2 + \epsilon^2_1 + 2 \epsilon_1 |\hat{h}_{\mathrm{SI},k}|,\label{SI_Defa}\\
 \|\mathbf{H}_{\mathrm{SI},\mathrm{BS}} \mathbf{v}_{k} \|^2  \!\!\!& \leq &\!\!\! \| {\mathbf{H}}_{\mathrm{SI},\mathrm{BS}} \|^2 \| \mathbf{v}_{k} \|^2 \nonumber\\
 \!\!\!& = &\!\!\! \| \hat{\mathbf{H}}_{\mathrm{SI},\mathrm{BS}} + \triangle{\mathbf{H}}_{\mathrm{SI},\mathrm{BS}}\|^2 \| \mathbf{v}_{k}\|^2 \nonumber\\
 \!\!\!& \leq &\!\!\! (\| \hat{\mathbf{H}}_{\mathrm{SI},\mathrm{BS}}  \|  + \| \triangle{\mathbf{H}}_{\mathrm{SI},\mathrm{BS}} \|)^2 \| \mathbf{v}_{k}\|^2 \nonumber\\
  \!\!\!& {\textcolor{red}{\leq}} &\!\!\!  (\| \hat{\mathbf{H}}_{\mathrm{SI},\mathrm{BS}}  \|^2 + \epsilon^2_2 + 2  \| \hat{\mathbf{H}}_{\mathrm{SI},\mathrm{BS}}  \| \epsilon_2) \|\mathbf{v}_{k} \|^2. \nonumber\\
\label{SI_Defb}
\end{eqnarray}
\end{subequations}
Note that $\epsilon_k$ is the minimal knowledge of the upper-bound of the channel error which is sufficient enough to describe the error in the absence of statistical information about the error. As a result, from (\ref{SI_Defa}) and (\ref{SI_Defb}), we obtain
 \begin{subequations}
\begin{align}
 \max_{|\triangle h_{\mathrm{SI}}| \leq \epsilon_1} &~ |{h}_{\mathrm{SI},k}|^2   \leq  |\hat{h}_{\mathrm{SI},k}|^2 + \epsilon^2_1 + 2 \epsilon_1 |\hat{h}_{\mathrm{SI},k}|,\label{mincd}\\
  \max_{\| \triangle \mathbf{{H}}_{\mathrm{SI},\mathrm{BS}} \mathbf{v}_k\| \leq \epsilon_2} &~ \|\mathbf{H}_{\mathrm{SI},\mathrm{BS}} \mathbf{v}_k \|^2
   \leq (\| \hat{\mathbf{H}}_{\mathrm{SI},\mathrm{BS}}  \|^2 + \epsilon^2_2 \nonumber\\
   &\qquad\qquad+ 2  \| \hat{\mathbf{H}}_{\mathrm{SI},\mathrm{BS}}  \| \epsilon_2) \|\mathbf{v}_k \|^2. \label{mincdd}
 \end{align}
\end{subequations}
On the other hand, it holds that
\begin{eqnarray}\label{SI2_Def}
  | (\hat{h}_{\mathrm{SI},k} + \triangle{h}_{\mathrm{SI},k})|^2   \!\!\!& \geq &\!\!\!(| \hat{h}_{\mathrm{SI},k}| - |\triangle{h}_{\mathrm{SI},k}|)^2 \nonumber\\
  \!\!\!& \geq &\!\!\! | \hat{h}_{\mathrm{SI},k}|^2 + \epsilon^2_1 - 2 | \hat{h}_{\mathrm{SI},k}|^2 \epsilon_1.
\end{eqnarray}
Here, we assume that  $|\hat{h}_{\mathrm{SI}}| \geq |\triangle{h}_{\mathrm{SI}}|$ which essentially means that the error $|\triangle{h}_{\mathrm{SI}}|$ is sufficiently small in comparison to the estimate or the estimate is meaningful. Accordingly,
 \begin{equation}\label{mincd2}
 \min_{|\triangle h_{\mathrm{SI},k}| \leq \epsilon_1} |{h}_{\mathrm{SI},k}|^2 \geq | \hat{h}_{\mathrm{SI},k}|^2 + \epsilon^2_1 - 2 | \hat{h}_{\mathrm{SI},k}| \epsilon_1.
 \end{equation}

\subsection{Problem Formulation}
We assume that each ${\rm MS}_k$ is characterized with strict QoS constraints. The QoS constraints require that the SINR for the downlink channel should be higher than a given threshold denoted by $\gamma^{\mathrm{{DL}}}_k,$ at all times in order to ensure a continuous information transfer. Similarly, each ${\rm MS}_k$ also requires that its harvested power must be above certain useful level specified by a prescribed threshold denoted by $\bar{Q}_k$ in order to maintain its receiver's operation. Meanwhile, for the uplink channel, each ${\rm MS}_k$ is expected to send feedback to the BS; thus a strict QoS is required such that the SINR of the uplink channel for each ${\rm MS}_k$ is expected to be no less than a given threshold denoted as $\gamma^{\mathrm{{UL}}}_{{k}}.$ It is worth noting that FD brings SI to the BS and ${\rm MS}_k$, and thus both the BS and ${\rm MS}_k$ may not always use their maximum transmit power as it increases the level of RSI. The BS and ${\rm MS}_k$ must therefore carefully choose their transmit power. Considering the above constraints, our objective is to minimize the end-to-end sum transmit power for the MISO FD SWIPT system by jointly designing the transmit beamforming vector (${\mathbf{v}}_k$) at the BS, the transmit power $P_{{\rm up},k}$ and the receiver PS ratio, ($\rho_k$), at ${\rm MS}_k$. Hence, the problem can be formulated as shown in (\ref{probss}) (at the top of the next page).
\begin{figure*}[!t]
\begin{eqnarray}
\!\!\!& &\!\!\! \min_{{\bf v}_k, {\bf w}_k, P_{{\rm up},k} ,\rho_k} \sum^K_{k=1}  (\|{\bf v}_k\|^2 + P_{{\rm up},k}) \nonumber\\
\!\!\!& &\!\!\! \mbox{s.t.}\nonumber\\
\!\!\!& &\!\!\! \min_{\|\triangle {\bf H}_{{\rm SI},{\rm BS}} \|\leq \epsilon_2} \frac{P_{{\rm up},k} | {\bf w}^H_k {\bf h}_{{\rm ul},k}|^2}
{  \sum^K_{j \neq k} P_{{\rm up},j} |{\bf w}^H_k {\bf h}_{{\rm ul},j}|^2 +  \sum^K_{j=1} \|{{{\mathbf{H}}}}_{\mathrm{SI},\mathrm{BS}} {\mathbf{v}}_{j}  \|^2 \| \mathbf{w}_k\|^2   + \| \mathbf{w}_k\|^2}
\geq \gamma^{\mathrm{{UL}}}_{{k}}, \forall k,\nonumber\\
\!\!\!& &\!\!\! \min_{|\triangle{h}_{\mathrm{SI},k}| \leq \epsilon_1} \frac{\rho_k|{\mathbf{h}}^H_{{\rm dl},k} {\mathbf{v}}_k|^2}{\rho_k( \sum_{j \neq k} |{\mathbf{h}}^H_{{\rm dl},k} {\mathbf{v}}_j|^2  +  |{{{h}}}_{\mathrm{SI},k}|^2 P_{{\rm up},k}  + \sigma^2_k ) + \delta^2_k}\geq \gamma^{\mathrm{{DL}}}_k,\forall k \nonumber\\
\!\!\!& &\!\!\! \min_{|\triangle{h}_{\mathrm{SI},k}| \leq \epsilon_1} (1-\rho_k) \left( \sum^K_{j = 1}|{\mathbf{h}}^H_k {\mathbf{v}}_j|^2 + |{{{h}}}_{\mathrm{SI},k} |^2 P_{{\rm up},k}  + \sigma^2_k \right)\geq \bar{Q}_{k},\forall k \nonumber\\
\!\!\!& &\!\!\! 
0 < P_{{\rm up},k} \leq \mathrm{min} (\bar{Q}_k, P_{max}),\quad  0 < \|\mathbf{v}_k\|^2 \leq P_{max},  \forall k, \nonumber\\
\!\!\!& &\!\!\!  0 < \rho_k < 1, \forall k 
 \label{probss}
\end{eqnarray}
\hrulefill
\vspace*{4pt}
\end{figure*}
Substituting the result obtained in (\ref{mincd}), (\ref{mincdd}) and  (\ref{mincd2}) into (\ref{probss}), the optimization problem in (\ref{probss}) can now be upper-bounded as given in (\ref{probs}) (at the top of the next page).
\begin{figure*}[!t]
\begin{eqnarray}
\!\!\!& &\!\!\! \min_{{{\mathbf{v}}_k, \mathbf{w}_k, P_{{\rm up},k} ,\rho_k}} \sum^K_{k=1}  (\|{\mathbf{v}}_k\|^2 + P_{{\rm up},k}) \nonumber\\
\!\!\!& &\!\!\! \mathrm{ s.t.}\nonumber\\
\!\!\!& &\!\!\! \frac{ \sum^K_{k=1} |\mathbf{w}^{H}_k {{\mathbf{h}}}_{{\rm ul},k} |^2 P_{{\rm up},k} }{\sum^K_{j \neq k} P_{{\rm up},j} |\mathbf{w}^H_k \mathbf{h}_{{\rm ul},j}|^2 +  (\| \hat{\mathbf{H}}_{\mathrm{SI},\mathrm{BS}}  \|^2 + \epsilon^2_2 + 2  \| \hat{\mathbf{H}}_{\mathrm{SI},\mathrm{BS}}  \|^2 \epsilon_2)K P_{max} \|\mathbf{w}_k \|^2  + \|\mathbf{w}_k \|^2} \geq \gamma^{\mathrm{{UL}}}_{{k}}, \forall k,\nonumber\\
\!\!\!& &\!\!\! \frac{\rho_k|{\mathbf{h}}^H_{{\rm dl},k} {\mathbf{v}}_k|^2}{\rho_k( \sum_{j \neq k} |{\mathbf{h}}^H_{{\rm dl},k} {\mathbf{v}}_j|^2  +  (|\hat{h}_{\mathrm{SI},k}|^2 + \epsilon^2_1 + 2 \epsilon_1 |\hat{h}_{\mathrm{SI},k}|^2) P_{max}  + \sigma^2_k ) + \delta^2_k}\geq \gamma^{\mathrm{{DL}}}_k,\forall k \nonumber\\
\!\!\!& &\!\!\! (1-\rho_k) \left( \sum^K_{j = 1}|{\mathbf{h}}^H_k {\mathbf{v}}_j|^2 + (| \hat{h}_{\mathrm{SI},k}|^2 + \epsilon^2_1 - 2 | \hat{h}_{\mathrm{SI},k}|^2 \epsilon_1)  P_{max} + \sigma^2_k \right)\geq \bar{Q}_{k},\forall k \nonumber\\
\!\!\!& &\!\!\!  
0 < P_{{\rm up},k} \leq \mathrm{min} (\bar{Q}_k, P_{max}), \quad  0< \|\mathbf{v}_k\|^2 \leq P_{max}, \forall k,  \nonumber\\
\!\!\!& &\!\!\!  0 < \rho_k < 1, \forall k  
 \label{probs}
\end{eqnarray}
\hrulefill
\vspace*{4pt}
\end{figure*}
Note that
the upper bound of the SI at the BS and ${\rm MS}_k$ is obtained when the source node transmits at maximum available power, i.e., when $P_{{\rm up},k} = \| \mathbf{v}_{k}\|^2 = P_{max}$ \cite{uplink}. As such, we denote the upper-bound of the SI power at the BS and ${\rm MS}_k$  as $\bar{E}$ and $\bar{G},$ respectively. Therefore, (\ref{probs}) is rewritten as
\begin{eqnarray}
\!\!\!& &\!\!\! \min_{{{\mathbf{v}}_k, \mathbf{w}_k,  P_{{\rm up},k},\rho_k}} \sum^K_{k=1}  (\|{\mathbf{v}}_k\|^2 + P_{{\rm up},k})  \nonumber\\
\!\!\!& &\!\!\! \mathrm{ s.t.}\nonumber\\
\!\!\!& &\!\!\! \frac{  P_{{\rm up},k} | \mathbf{w}^H_k {{\mathbf{h}}}_{{\rm ul},k} \hspace{0.3mm}|^2}{  \sum^K_{j \neq k} P_{{\rm up},j} |\mathbf{w}^H_k \mathbf{h}_{{\rm ul},j}|^2 +   \bar{E} \|\mathbf{w}_k\|^2   + \| \mathbf{w}_k\|^2}
 \geq \gamma^{\mathrm{{UL}}}_{{k}}, \forall k,\nonumber\\
\!\!\!& &\!\!\! \frac{\rho_k|{\mathbf{h}}^H_{{\rm dl},k} {\mathbf{v}}_k|^2}{\rho_k( \sum_{j \neq k} |{\mathbf{h}}^H_{{\rm dl},k} {\mathbf{v}}_j|^2  + \bar{G}_k + \sigma^2_k ) + \delta^2_k}
  \geq \gamma^{\mathrm{{DL}}}_k,\forall k \nonumber\\
\!\!\!& &\!\!\! (1-\rho_k) \left( \sum^K_{j = 1}|{\mathbf{h}}^H_k {\mathbf{v}}_j|^2 + \tilde{G}_k +     \sigma^2_k \right)
  \geq \bar{Q}_k,\forall k \nonumber\\
\!\!\!& &\!\!\! 
 0 < P_{{\rm up},k} \leq \mathrm{min} (\bar{Q}_k, P_{max}),~~0 < \|\mathbf{v}_k\|^2 \leq P_{max}, \forall k \nonumber\\
\!\!\!& &\!\!\!  0 < \rho_k < 1,
 \label{probs*a}
\end{eqnarray}
where  $\bar{E} \triangleq (\| \hat{\mathbf{H}}_{\mathrm{SI},\mathrm{BS}}  \|^2 + \epsilon^2_2 + 2  \| \hat{\mathbf{H}}_{\mathrm{SI},\mathrm{BS}}  \|^2 \epsilon_2)KP_{max},$ $\bar{G}_k \triangleq (|\hat{h}_{\mathrm{SI},k}|^2 + \epsilon^2_1 + 2 \epsilon_1 |\hat{h}_{\mathrm{SI},k}|^2)  {P}_{max} $ and $\tilde{G}_k \triangleq ( | \hat{h}_{\mathrm{SI},k}|^2 + \epsilon^2_1 - 2 | \hat{h}_{\mathrm{SI},k}|^2 \epsilon_1) {P}_{max}$ denotes the maximum SI power associated with the energy harvesting constraint at ${\rm MS}_k$.

In this paper, we investigate the general case where all MSs are characterized as having a non-zero SINR and harvested power targets, i.e., $\gamma^{\mathrm{DL}}_k,\gamma^{\mathrm{UL}}_k,\bar{Q}_k> 0, \forall k$. As such, the receive PS ratio at all MSs should satisfy $0 < \rho_k < 1,$ as given by the PS ratio constraint. It is easy to see that (\ref{probs*a}) is non-convex and hard to solve. Thus, we solve this problem in a two-step process. Firstly, we observe that the QoS uplink constraint ($\gamma^{\mathrm{{UL}}}_{{k}}$) does not have the PS coefficient and this is because in our model, the BS is not designed to harvest energy. Hence, we can decompose problem (\ref{probs*a}) into two sub-problems. The resulting sub-problems can be written as
\begin{eqnarray}
\!\!\!& &\!\!\! \min_{ \mathbf{w}_k, P_{{\rm up},k}} \sum^K_{k=1}  P_{{\rm up},k} \nonumber\\
\!\!\!& &\!\!\! \mathrm{ s.t.}\nonumber\\
\!\!\!& &\!\!\! \frac{  P_{{\rm up},k} | \mathbf{w}^H_k {{\mathbf{h}}}_{{\rm ul},k} |^2}{  \sum^K_{j \neq k} P_{{\rm up},j} |\mathbf{w}^H_k \mathbf{h}_{{\rm ul},j}|^2 +   \bar{E} \|\mathbf{w}_k\|^2   + \| \mathbf{w}_k\|^2}
 \geq \gamma^{\mathrm{{UL}}}_{{k}}, \forall k, \nonumber\\
\!\!\!& &\!\!\! 0 < P_{{\rm up},k} \leq \mathrm{min} (\bar{Q}_k, P_{max}), \forall k,
 \label{probs2}
\end{eqnarray}
and
\begin{eqnarray}
\!\!\!& &\!\!\! \min_{{{\mathbf{v}}_k,\rho_k}} \sum^K_{k=1}  \|{\mathbf{v}}_k\|^2 \nonumber\\
\!\!\!& &\!\!\! \mathrm{ s.t.}\nonumber\\
\!\!\!& &\!\!\! \frac{\rho_k|{\mathbf{h}}^H_{{\rm dl},k} {\mathbf{v}}_k|^2}{\rho_k\left( \sum_{j \neq k} |{\mathbf{h}}^H_{{\rm dl},k} {\mathbf{v}}_j|^2  + \bar{G}_k + \sigma^2_k \right) + \delta^2_k}
  \geq \gamma^{\mathrm{{DL}}}_k,\forall k, \nonumber\\
\!\!\!& &\!\!\! (1-\rho_k) \left( \sum^K_{j = 1}|{\mathbf{h}}^H_k {\mathbf{v}}_j|^2 + \tilde{G}_k + \sigma^2_k \right)
  \geq \bar{Q}_{k},\forall k \nonumber\\
\!\!\!& &\!\!\! 
0 < \|\mathbf{v}_k\|^2 \leq P_{max}, \forall k, \nonumber\\
\!\!\!& &\!\!\!  0 < \rho_k < 1, \forall k.
 \label{probs3}
\end{eqnarray}
Note that \eqref{probs2} corresponds to optimizing the variables involved in the uplink, and \eqref{probs3} involves those in the downlink. Next, we apply SDR to the sub-problems as discussed below.

 \section{Solutions}\label{opt_soltn}
In this section, we will focus on how to solve (\ref{probs2}) and (\ref{probs3}) optimally. We start by solving (\ref{probs2}) to determine the optimal value $P^*_{{\rm up},k}$ and $\mathbf{w}^*_k.$ For given $\mathbf{w}_k,$ the optimal $P^*_{{\rm up},k}$ can be determined. Problem (\ref{probs2}) is thus reformulated as
\begin{subequations}
\begin{eqnarray}
\!\!\!& &\!\!\! \min_{ P_{{\rm up},k}} \sum^K_{k=1}     P_{{\rm up},k} \label{probs2*&} \\
\!\!\!& &\!\!\! \mathrm{ s.t.}\nonumber\\
\!\!\!& &\!\!\! \frac{  P_{{\rm up},k} | \mathbf{w}^H_k {{\mathbf{h}}}_{{\rm ul},k} |^2}{  \sum^K_{j \neq k}\! P_{{\rm up},j} |\mathbf{w}^H_k \mathbf{h}_{{\rm ul},j}|^2 \!+ \!\!  \bar{E} \|\mathbf{w}_k\|^2  \!\! + \! \| \mathbf{w}_k\|^2}
\!\geq \! \gamma^{\mathrm{{UL}}}_{{k}}\!\!,\label{probs2*&1}\\
\!\!\!& &\!\!\! 0 < P_{{\rm up},k} \leq \mathrm{min} (\bar{Q}_k, P_{max}), \label{probs2*1&} \forall k.
\end{eqnarray}
\end{subequations}
The optimal $P^*_{{\rm up},k}$ is the minimum $P_{{\rm up},k}$ which satisfies (\ref{probs2*&1}) to equality. As a result, the optimal $P_{{\rm up},k}$ is given by
\begin{eqnarray}
P^*_{{\rm up}} = \frac{\gamma^{\mathrm{{UL}}}_{{k}} (  \bar{E} \|\mathbf{w}_k\|^2   + \| \mathbf{w}_k\|^2 )}{ | \mathbf{w}^H_k {{\mathbf{h}}}_{{\rm ul},k} |^2 - \gamma^{\rm UL}_k (\sum^K_{j \neq k} | \mathbf{w}^H_k {{\mathbf{h}}}_{{\rm ul},j}|^2)}.
\end{eqnarray}

The optimal receiver can be defined as the Wiener filter \cite{uplink}
\begin{eqnarray}
\mathbf{w}^*_k \!\!\!&=&\!\!\! \left( \sum^K_{j=1} P_{{\rm up},j} \mathbf{h}_{{\rm ul},j} \mathbf{h}^H_{{\rm ul},j} + \left[ \sigma^2_j + \sum^K_{j=1} \| \mathbf{v}_j \|^2  \right] \mathbf{I}  \right)^{-1} \!\!\times \nonumber\\
\!\!\!& &\!\!\! \sqrt{P_{{\rm up},j}} \mathbf{h}_{{\rm ul},j}.
\end{eqnarray}
Secondly, we investigate problem (\ref{probs3}) to determine the optimal value of the receive PS ratio and the transmit beamforming vector at the BS. It is worth pointing out that the feasibility of problem (\ref{probs3}) has been proved in \cite{joint_transmit}. Accordingly, by applying SDP technique to solve problem (\ref{probs3}), we define $\mathbf{Z}_k = {\mathbf{v}}_k {\mathbf{v}}^{H}_k, \forall k.$ Thus, it follows that ${\rm Rank}(\mathbf{Z}_k)\leq1, \forall k.$ If we ignore the rank-one constraint for all ${\mathbf{Z}_{k}}$'s, the SDR of problem (\ref{probs3}) can be written as
\begin{eqnarray}
\!\!\!& &\!\!\! \min_{{\mathbf{Z}_k,\rho_k}} \sum^K_{k=1}   \mathrm{Tr}(\mathbf{Z}_k)  \nonumber\\
\!\!\!& &\!\!\! \mathrm{ s.t.}\nonumber\\
\!\!\!& &\!\!\! \frac{\rho_k {\mathbf{h}}^H_{{\rm dl},k} \mathbf{Z}_k {\mathbf{h}}_{{\rm dl},k} }{\rho_k( \sum_{j \neq k} {\mathbf{h}}^H_{{\rm dl},k} \mathbf{Z}_j {\mathbf{h}}_k  + \bar{G}_k + \sigma^2_k ) + \delta^2_k}\geq \gamma^{\mathrm{{DL}}}_k,\forall k, \nonumber\\
\!\!\!& &\!\!\! (1-\rho_k) \left( \sum^K_{j = 1}{\mathbf{h}}^H_{{\rm dl},k} \mathbf{Z}_j {\mathbf{h}}_{{\rm dl},k} + \tilde{G}_k + \sigma^2_k \right)\geq \bar{Q}_{k},\forall k, \nonumber\\
\!\!\!& &\!\!\!  0 < \rho_k < 1,\forall k \nonumber\\
\!\!\!& &\!\!\! \mathbf{Z}_k \succeq 0,  \forall k.
 \label{probs6}
\end{eqnarray}
Problem (\ref{probs6}) is non-convex since both the SINR and harvested power constraints involve coupled $\mathbf{Z}_k$ and $\rho_k$'s. Nonetheless,  (\ref{probs6}) can be reformulated as the following problem:
\begin{eqnarray}
\!\!\!& &\!\!\! \min_{{\{\mathbf{Z}_k,\rho_k}\}} \sum^K_{k=1}   \mathrm{Tr}(\mathbf{Z}_k)  \nonumber\\
\!\!\!& &\!\!\! \frac{1}{\gamma^{\mathrm{{DL}}}_k} {\mathbf{h}}^H_{{\rm dl},k} \mathbf{Z}_k {\mathbf{h}}_{{\rm dl},k} -  \sum_{j \neq k} {\mathbf{h}}^H_{{\rm dl},k} \mathbf{Z}_j {\mathbf{h}}_{{\rm dl},k}  + \bar{G}_k \geq \sigma^2_k + \frac{\delta^2_k}{\rho_k}, \forall k, \nonumber\\
\!\!\!& &\!\!\!  \sum^K_{j = 1}{\mathbf{h}}^H_{{\rm dl},k} \mathbf{Z}_j {\mathbf{h}}_{{\rm dl},k} +  \tilde{G}_k \geq \frac{\bar{Q}_{k}}{(1-\rho_k)} - \sigma^2_k, \forall k,\nonumber\\
\!\!\!& &\!\!\!  0 < \rho_k < 1,\forall k,\nonumber\\
\!\!\!& &\!\!\! \mathbf{Z}_k \succeq 0, \forall k.
\label{probs7}
\end{eqnarray}
As shown in (\ref{probs7}), both $\frac{1}{\rho_k}$ and $\frac{1}{1-\rho_k}$ are convex functions over $\rho_k,$ thus problem (\ref{probs7}) is convex and can be solved using disciplined convex programming (CVX). To proceed, let $\mathbf{Z}^*_k$ denote the optimal solution to problem (\ref{probs7}). Accordingly, it follows that if $\mathbf{Z}^*_k$ satisfies ${\rm Rank}(\mathbf{Z}^*_k)=1, \forall k,$ then the optimal beamforming solution ${\mathbf{v}}^*_k$ to problem (\ref{probs3}) can be obtained from the eigenvalue decomposition of $\mathbf{Z}^*_k$, for $k = 1, \dots, K$ and the optimal PS solution of problem (\ref{probs3}) is given by the associated $\rho^*_k$'s.  However, in the case that there exists any $k$ such that ${\rm Rank}(\mathbf{Z}^*_k)>1,$ then in general the solution $\mathbf{Z}^*_k$ and $\rho^*_k$ of problem (\ref{probs7}) is not always optimal for problem (\ref{probs3}). We show in the appendix that it is indeed true that for problem (\ref{probs3}), the solution satisfies ${\rm Rank}(\mathbf{Z}^*_k)=1, \forall k.$
\begin{proposition} \label{prop_1}
Given $\gamma^{\mathrm{{DL}}}_k > 0$ and $\bar{Q}_{k}> 0, \forall k$, for (\ref{probs7}), we have
\begin{enumerate}
\item $\{\mathbf{Z}^*_k\}$ and $\{\rho_k\}$ satisfy the first two sets of constraints of (\ref{probs7}) with equality;
\item $\{\mathbf{Z}^*_k\}$ satisfies ${\rm Rank}(\mathbf{Z}^*_k) =1, \forall k$.
\end{enumerate}
\end{proposition}


\begin{proof}
Please refer to Appendix A.
\end{proof}

\section{Suboptimal Solution}\label{sub}
To effectively make meaningful comparison based on the performance analysis for SWIPT in MISO FD systems, in this section, we investigate a suboptimal solution based on ZF by jointly designing the beamforming vector and PS ratios.

\subsection{ZF Beamforming}\label{zero_forcing}
To simplify the beamforming design, we add the ZF constraint. As such, by restricting ${\mathbf{v}}_k$ in (\ref{probs3}) to satisfy ${\mathbf{h}}^H_{{\rm dl},i} {\mathbf{v}}_k=0, \forall i \neq k,$ ZF can be used to eliminate multiuser interference.
Applying the ZF transmit beamforming constraint, (\ref{probs3}) can be reformulated as
\begin{eqnarray}
\!\!\!& &\!\!\! \min_{{\{{\mathbf{v}}_k,\rho_k\}}} \sum^K_{k=1}  \|{\mathbf{v}}_k\|^2 \nonumber\\
\!\!\!& &\!\!\! \mathrm{ s.t.}\nonumber\\
\!\!\!& &\!\!\! \frac{\rho_k|{\mathbf{h}}^H_{{\rm dl},k} {\mathbf{v}}_k|^2}{\rho_k(\bar{G}_k + \sigma^2_k ) + \delta^2_k}
 \geq \gamma^{\mathrm{{DL}}}_k,\forall k, \nonumber\\
\!\!\!& &\!\!\! (1-\rho_k) \left( |{\mathbf{h}}^H_{{\rm dl},k} {\mathbf{v}}_k|^2 + \tilde{G}_k + \sigma^2_k \right)
  \geq \bar{Q}_{k},\forall k \nonumber\\
\!\!\!& &\!\!\! \mathbf{H}^H_{{\rm dl},k} {\mathbf{v}}_k = 0,  \|\mathbf{v}_k\|^2 \leq P_{max}, \forall k, \nonumber\\
\!\!\!& &\!\!\!  0 < \rho_k < 1,\forall k,
 \label{probs8}
\end{eqnarray}
in which $\mathbf{H}_{{\rm dl},k} \triangleq [\mathbf{h}_{{\rm dl},1} \cdots \mathbf{h}_{{\rm dl},k-1}, \mathbf{h}_{{\rm dl},k+1} \cdots  \mathbf{h}_{{\rm dl},K}] \in \mathbb{C}^{N_t \times (K-1)}.$ Clearly, problem (\ref{probs}) must be feasible if $N_t \geq K$ due to the ZF transmit beamforming\cite{joint_transmit}. Proposition \ref{prop^3} gives the optimal solution to problem (\ref{probs8}). 

\begin{proposition} \label{prop^3}
From the result in \cite{joint_transmit}, let $\mathbf{U}_k$ denote the orthogonal basis of the null space of $\mathbf{H}^H_{{\rm dl},k}, k = 1, \dots, K.$ The optimal solution to problem (\ref{probs8}) is thus given by
\begin{align}
 \tilde{\rho}^*_k &= {\frac{ + \beta_k \pm \sqrt{\beta^2_k + 4 \alpha_k C_k}}{2 \alpha_k}},  \forall k,\\
\tilde{{\mathbf{v}}}^*_k &= \sqrt{\gamma^{\mathrm{{DL}}}_k \left(\bar{G}_k + \sigma^2_k + \frac{\delta^2_k}{\rho_k}\right)} \frac{\mathbf{U}_k \mathbf{U}^H_k {\mathbf{h}}_{{\rm dl},k}}{\|\mathbf{U}_k \mathbf{U}^H_k {\mathbf{h}}_{{\rm dl},k} \|^2}, \forall k.
\end{align}
\end{proposition}

\begin{proof}
Please refer to Appendix B.
\end{proof}

\section{Simulation Results}\label{N_exam}
Here, we investigate the performance of the proposed joint beamforming and received PS (JBPS) optimization design for SWIPT in MISO FD systems through computer simulations. In particular, we simulated a flat Rayleigh fading environment in which the channel fading coefficients are characterized as complex Gaussian numbers with zero mean and are i.i.d.~and we assume there are  $K = 2$ MSs and all MSs have the same set of parameters i.e., $\sigma^2_k = \sigma^2,$   $\delta^2_k = \delta^2,$ $\bar{Q}_k = Q,$ and  $\gamma^{\mathrm{{DL}}}_k = \gamma^{\mathrm{{DL}}}.$ We also assume that $60\%$ of the SI power has been cancelled using existing SIC techniques \cite{duarte}. All simulations are averaged over $500$ independent channel realizations.
\begin{figure}[ht]
\centering
\includegraphics*[width=8cm]{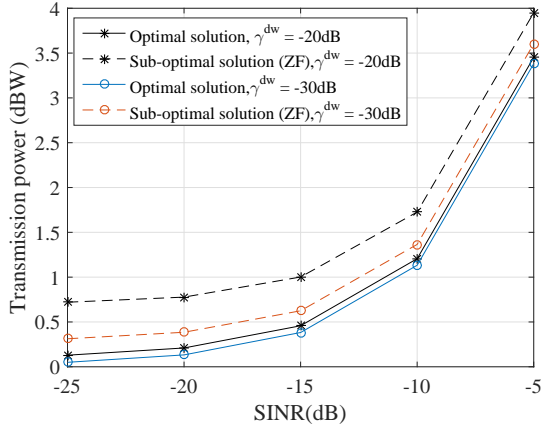}
\caption{Transmission power versus SINR, $\gamma^{\mathrm{{UL}}}$.}\label{fig_2}
\end{figure}

In Fig.~\ref{fig_2}, we investigate the minimum end-to-end transmission power for SWIPT in MISO FD systems versus the SINR target for all MSs, $\gamma^{\mathrm{{UL}}},$ for fixed harvested power threshold $Q = 20$ dBm. It is assumed that the BS is equipped with $N_t = 2$ transmit antennas. Fig.~\ref{fig_2} shows the performance comparison in terms of end-to-end sum transmit power, between the optimal JBPS solution to (\ref{probs*a}) and the suboptimal solution based on ZF beamforming. As can be observed, the minimum end-to-end sum transmit power rises with the increase in $\gamma^{\mathrm{{UL}}}.$ However, for different values of $\gamma^{\mathrm{{DL}}}$, the optimal JBPS scheme outperforms the optimization scheme based on ZF beamforming. For example, at $\gamma^{\mathrm{{DL}}} = -20$dB,  the optimal JBPS scheme achieves a near 1dB gain over the suboptimal ZF beamforming scheme. It is also observed that for both cases of $\gamma^{\mathrm{{DL}}} = -20$dB and $\gamma^{\mathrm{{DL}}} = -30$dB, the minimum end-to end-transmission power is achieved by optimal JBPS solution for all values of  $\gamma^{\mathrm{{UL}}}$. Thus, with an increase in SINR uplink threshold, $\gamma^{\mathrm{UL}},$  the optimal JBPS scheme achieves a transmit power gain over the suboptimal ZF beamforming scheme.

 \begin{figure}[ht]
\centering
\includegraphics*[width=8cm]{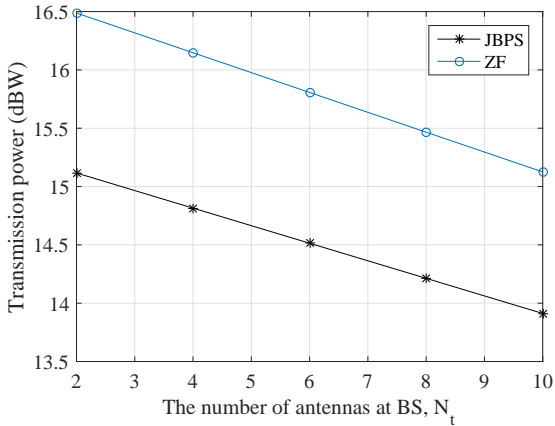}
\caption{Transmission power versus number of transmit antenna at BS, N$_t$.}\label{fig_3}
\end{figure}

In Fig.~\ref{fig_3}, we study the impact of the number of transmit antennas at the BS, $N_t,$ on the minimum end-to-end transmission power for the proposed solutions for fixed harvested power threshold, $Q = 20$dBm. As can be observed, the minimum end-to-end sum transmit power decreases with the increase in the number of the transmit antennas at the BS. However, the optimal JBPS scheme outperforms the optimization scheme based on ZF beamforming. For example, for $N_t = 2$, $\gamma^{\mathrm{{DL}}} = -20$dB and $\gamma^{\mathrm{{UL}}} = -20$dB, the optimal JBPS achieves 1dB gain over the suboptimal ZF beamforming scheme. Thus, we can conclude that more transmit antennas at the BS which adopts beamforming allow it to focus more power to ${\rm MS}_k.$
\begin{figure}[ht]
\centering
\includegraphics*[width=8cm]{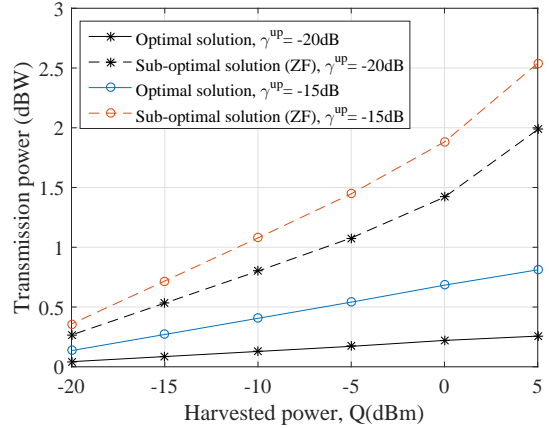}
\caption{Transmission power versus harvested energy.}\label{fig_4}
\end{figure}

In Fig.~\ref{fig_4}, we illustrate the minimum transmission power achieved by JBPS and ZF for a downlink SINR,  $\gamma^{\mathrm{DL}} = -20$dB,  for different threshold of the harvested power. As observed in Fig.~\ref{fig_4}, the optimal JBPS schemes achieves the minimum transmission power for all values of the harvested power threshold. Also, the increased harvested energy threshold demands more transmit power. We also see that for increasing values of the harvested power threshold, JBPS  achieves an increasing transmit power gain over the ZF scheme.

\section{Conclusion}\label{conc}
This paper investigated the joint transmit beamforming and receive PS design for SWIPT in MISO FD system. The end-to-end sum transmit power has been minimized subject to the given SINR and harvested power constraints for each MS by jointly optimizing the transmit beamforming vector at the BS, the PS ratio and the transmit power at the MSs. A suboptimal scheme based on ZF was also presented. We showed through simulation results that the proposed optimal scheme achieves a transmit power gain over the suboptimal ZF scheme.

\section*{Appendix A}\label{app_B}
\section*{Proof of Proposition \ref{prop_1} }\label{Proof_2}
Firstly, let us proceed to prove the first part of Proposition \ref{prop_1}. Problem (\ref{probs7}) is convex and satisfies the Slater's condition, and therefore its duality gap is zero \cite{boyd}. We denote $\{\lambda_k\}$ and $\{\mu_k\}$ as the dual variables associated with the SINR constraints and harvested power constraints of problem (\ref{probs7}), respectively. The partial Lagrangian of problem (\ref{probs7}) is thus given as shown in (\ref{alex_1}) (top of next page).
\begin{figure*}[!t]
\begin{eqnarray}
{\cal  L}(\{\mathbf{Z}_k, \rho_k, \lambda_k, \mu_k \})  \!\!\!& \triangleq &\!\!\! \sum^K_{k=1}   \mathrm{Tr}(\textbf{Z}_k) -\sum^K_{k=1}   \lambda_k\left(  \frac{1}{\gamma^{\mathrm{{DL}}}_k} {\mathbf{h}}^H_{{\rm dl},k} \mathbf{Z}_k {\mathbf{h}}_{{\rm dl},k} -  \sum_{j \neq k} {\mathbf{h}}^H_{{\rm dl},k} \mathbf{Z}_j {\mathbf{h}}_{{\rm dl},k}  + \bar{G}_k - \sigma^2_k - \frac{\delta^2_k}{\rho_k}   \right)\nonumber\\
 \!\!\!&  - &\!\!\! \sum^K_{k=1}   \mu_k \left(\sum^K_{j = 1}\mathbf{\mathbf{h}}^H_{{\rm dl},k} \mathbf{Z}_j {\mathbf{h}}_{{\rm dl},k} +  \tilde{G}_k - \frac{\bar{Q}_k}{(1-\rho_k)} + \sigma^2_k \right)
\label{alex_1}
\end{eqnarray}
\hrulefill
\end{figure*}
Given the Lagrangian function, the dual function of problem (\ref{probs7}) is given by \cite[Sec.5.7.3]{boyd}
\begin{equation}
 \min_{{\mathbf{Z}_k \succeq 0, 0<\rho_k<1, \forall k}} {\cal L}(\{\mathbf{Z}_k, \rho_k, \lambda_k, \mu_k \}).
 \label{min^*}
\end{equation}
Equation  (\ref{min^*}) can explicitly be written as shown in (\ref{y6}) (at the top of the next page)
\begin{figure*}[!t]
\begin{align}
 \min_{{\mathbf{Z}_k \succeq 0, 0<\rho_k<1, \forall k}} \left\{ \sum^K_{k=1} \mathrm{Tr}(\mathbf{A}_k \mathbf{Z}_k) + \sum^K_{k=1} (-\lambda_k(\bar{G}_k -\sigma^2_k)- \mu_k(\tilde{G}_k + \sigma^2_k)) + \sum^K_{k=1} (\frac{\lambda_k \delta^2_k}{\rho_k}+ \frac{\mu_k \bar{Q}_k}{(1-\rho_k)})\right\},
\label{y6}
\end{align}
\hrulefill
\end{figure*}
where
\begin{equation}
\mathbf{A}_k = \mathbf{I}_{N_t} + \sum^K_{j=1}(\lambda_j - \mu_j) {\mathbf{h}}_{{\rm dl},j} {\mathbf{h}}^H_{{\rm dl},j}-\left(\frac{\lambda_k}{\gamma^{\mathrm{DL}}_k } + \lambda_k \right) {\mathbf{h}}_{{\rm dl},k} {\mathbf{h}}^H_{{\rm dl},k}.
\end{equation}
Denote $\{\lambda^*_k\}$ and $\{\mu^*_k\}$ as the optimal dual solution to problem (\ref{probs7}). As a result, we define
\begin{equation}
\mathbf{A}^*_k = \mathbf{I}_{N_t} + \sum^K_{j=1}(\lambda^*_j - \mu^*_j) {\mathbf{h}}_{{\rm dl},j} {\mathbf{h}}^H_{{\rm dl},j}-\left(\frac{\lambda^*_k}{\gamma^{\mathrm{DL}}_k } + \lambda^*_k \right) {\mathbf{h}}_{{\rm dl},k} {\mathbf{h}}^H_{{\rm dl},k}.
\end{equation}
We observe from (\ref{y6}) that for any given $k,$ $\mathbf{Z}^*_k$ must be a solution to the following problem
\begin{equation}
\min_{{\mathbf{Z}_k \succeq 0}}  \mathrm{Tr}(\mathbf{A}^*_k \mathbf{Z}_k).\label{prob_29}
\end{equation}
To guarantee a bounded dual optimal value, we must have
\begin{equation}
\mathbf{A}^*_k \succeq 0, ~\mbox{for }k = 1, 2,  \dots, K.
\end{equation}
Consequently, the optimal value for problem (\ref{prob_29}) is zero, i.e.,   $\mathrm{Tr}(\mathbf{A}^*_k \mathbf{Z}_k) = 0, k = 1, 2,  \dots€€, K,$ which in conjunction with $\mathbf{A}^*_k \succeq 0$ and $\mathbf{Z}^*_k \succeq 0, k = 1, 2,  \dots,€€€ K,$ implies that
 \begin{equation}
 \mathbf{A}^*_k\mathbf{Z}^*_k = 0,~\mbox{for } k = 1, 2,  \dots, K.\label{imp_1}
 \end{equation}
Nonetheless, from (\ref{y6}), it is observed that the optimal PS solution $\rho^*_k$ for any given $k \in \{1,  \dots, K\}$ must be a solution of the following problem:
\begin{equation}\label{opt_rho}
\min_{{\rho_k}}  \frac{\lambda^*_k \delta^2_k}{\rho_k}+ \frac{\mu^*_k \bar{Q}_k}{(1-\rho_k)}~~\mbox{s.t.}~~
0 < \rho_k < 1.
\end{equation}
Note that we observe from (\ref{opt_rho}) that for the case when $\lambda^*_k=0$ and $\mu^*_k>0$, the optimal solution will be $\rho^*_k \rightarrow 0.$ Similarly, for the case when $\mu^*_k=0$ and $\lambda^*_k>0,$ the optimal solution is $\rho^*_k \rightarrow 1.$ Since $\bar{Q}_k > 0$ and $\gamma^{\mathrm{DL}}_k \!\!> \!\!0, \forall k, 0 < \rho_k < 1$ must hold for all $k$'s in (\ref{probs7}), the above two cases cannot be true. Consequently, we prove that $\lambda^*_k = 0$ and $\mu^*_k = 0$ cannot be true for any $k$ by contradiction. Let us assume there exist some $k$'s such that $\lambda^*_k = \mu^*_k = 0.$ We therefore define a set
\begin{equation}
\Theta\triangleq \{k|\lambda^*_k=0, \mu^*_k = 0,1 \leq k \leq K\},~\mbox{where }\Theta \neq \Phi.
\end{equation}
We also define
\begin{equation}\label{equ_B}
\mathbf{B}^* \triangleq \mathbf{I}_{N_t} + \sum_{j\notin \Theta}( \lambda^*_j-\mu^*_j) {\mathbf{h}}_{{\rm dl},j} {\mathbf{h}}^H_{{\rm dl},j}.
\end{equation}
Then $\mathbf{A}^*_k$ can be written as
\begin{equation}
\mathbf{A}^*_k = \left\{ \begin{array}{ll}
         \mathbf{B}^*, & \mbox{if $k \in \Theta$};\\
        \mathbf{B}^*- \left(\frac{\lambda^*_k}{\gamma^{\mathrm{DL}}_k}+\lambda^*_k\right) {\mathbf{h}}_{{\rm dl},k} {\mathbf{h}}^H_{{\rm dl},k}, & \mbox{otherwise}.\end{array} \right.
\label{matrix_equa}
\end{equation}
Since $\mathbf{A}^*_k \succeq 0$ and $-\left(\frac{\lambda^*_k}{\gamma^{\mathrm{DL}}_k}+\lambda^*_k\right) {\mathbf{h}}_{{\rm dl},k} {\mathbf{h}}^H_{{\rm dl},k}\preceq 0,$ consequently,
$\mathbf{B}^* \succeq 0.$
Let us proceed to show that $\mathbf{B}^*\succ 0$ by contradiction. Assuming the minimum eigenvalue of $\mathbf{B}^*$ is zero, consequently, there exists at least an ${\mathbf{x}} \neq 0$ such that ${\mathbf{x}}^H \mathbf{B}^*{\mathbf{x}}= 0.$ From equation (\ref{matrix_equa}), it follows that
\begin{equation}\label{Theta}
{\mathbf{x}}^H \mathbf{A}^*_k{\mathbf{x}}=-\left(\frac{\lambda^*_k}{\gamma^{\mathrm{DL}}_k}+\lambda^*_k\right) {\mathbf{x}}^H  {\mathbf{h}}_{{\rm dl},k} {\mathbf{h}}^H_{{\rm dl},k} {\mathbf{x}}\geq 0, k \notin \Theta.
\end{equation}
Notice that we have $\lambda^*_k>0$ if $k\notin\Theta.$ Accordingly, from (\ref{Theta}) we obtain $|{\mathbf{h}}^H_{{\rm dl},k} {\mathbf{x}}|^2 \leq 0,$ $k\notin \Theta.$ It follows that
\begin{equation}
{\mathbf{h}}^H_{{\rm dl},k} {\mathbf{x}}=0, k \notin \Theta.
\end{equation}
Conclusively, we have
\begin{eqnarray}
{\mathbf{x}}^H \mathbf{B}^*{{\mathbf{x}}} \!\!\!& = &\!\!\! {\mathbf{x}}^H \left( \mathbf{I}_{N_t} + \sum_{j\notin\Theta}(\lambda^*_j-\mu^*_j) {\mathbf{h}}_{{\rm dl},j} {\mathbf{h}}^H_{{\rm dl},j}\right) {\mathbf{x}}\nonumber\\
\!\!\!& = &\!\!\! {\mathbf{x}}^H {\mathbf{x}}>0,
\end{eqnarray}
which contradicts to ${\mathbf{x}}^H \mathbf{B}^*{\mathbf{x}}=0.$ Thus, we have $\mathbf{B}^* \succ 0,$ i.e., ${\rm Rank}(\mathbf{B}^*) = N_t.$ We can therefore deduce from (\ref{matrix_equa}) that ${\rm Rank}(\mathbf{A}^*_k)=N_t$ if $k\in\Theta.$ From (\ref{imp_1}), we have $\mathbf{Z}^*_k=0$ if $k \in \Theta.$ However, we can easily verify that $\mathbf{Z}^*_k=0$ cannot be optimal for  (\ref{probs7}). Appropriately, it must follow that $\Theta = \Phi,$ i.e., $\lambda_k=0$ and $\mu_k =0$ cannot be true for any $k.$ Interestingly, as we have previously shown that both cases of $\lambda^*_k = 0,\mu^*_k = 0$  and $\lambda^*_k > 0,\mu^*_k = 0$ cannot be true for any $k,$ it follows that $\lambda^*_k > 0,\mu^*_k > 0,  \forall k.$ In agreement to complementary slakeness \cite{boyd}, the first part of of Proposition \ref{prop_1} is thus proved.
Secondly, we proceed to prove the second part of Proposition \ref{prop_1}. Since $\Theta = \Phi,$ it follows that (\ref{equ_B}) and (\ref{matrix_equa}) reduces to
\begin{equation}\label{equ_B1}
\mathbf{A}^*_k = \mathbf{B}^* -\left(\frac{\lambda^*_k}{\gamma^{\mathrm{DL}}_k}+\lambda^*_k\right) {\mathbf{h}}_{{\rm dl},k} {\mathbf{h}}^H_{{\rm dl},k}, k=1,  \dots, K.
\end{equation}
On account of the fact that we have shown from the first part of the proof that ${\rm Rank}(\mathbf{B}^*)=N_t$, it follows that ${\rm Rank}(\mathbf{A}^*_k) \geq N_t - 1,$  $k=1,  \dots, K.$ Notice that if $\mathbf{A}^*_k$ is characterized as having a full rank, then we have $\mathbf{Z}^* = 0,$ which cannot be the optimal solution to (\ref{probs7}). Thus, it follows that ${\rm Rank}(\mathbf{A}^*_k) = N_t - 1, \forall k$. According to (\ref{imp_1}), we have ${\rm Rank}(\mathbf{Z})^* = 1, k = 1, \dots, K.$ We thus proved the second part of Proposition \ref{prop_1}. By combining the proofs for both parts, we have thus completed the proof of Proposition \ref{prop_1} \cite{joint_transmit}.

\section*{Appendix B}\label{app_C}
\section*{Proof of Proposition \ref{prop^3} }\label{proof_3}
From (\ref{probs8}), we see that the ZF transmit beamforming constraints make it possible for us to decouple the SINR and the harvested power constraints over $k$ because the objective function in problem  (\ref{probs8}) is separable over $k$. Therefore, problem (\ref{probs8}) can be decomposed into $K$ subproblems, $k = 1, \dots,K,$  with the $k$th subproblem expressed as
\begin{eqnarray}
\!\!\!& &\!\!\! \min_{{{\mathbf{v}}_k,\rho_k}}  \|{\mathbf{v}}_k\|^2 \nonumber\\
\!\!\!& &\!\!\! \mathrm{ s.t.}\nonumber\\
\!\!\!& &\!\!\! \frac{\rho_k|{\mathbf{h}}^H_{{\rm dl},k} {\mathbf{v}}_k|^2}{\rho_k(\bar{G}_k + \sigma^2_k ) + \delta^2_k}
 \geq \gamma^{\mathrm{{DL}}}_k, \nonumber\\
\!\!\!& &\!\!\! (1-\rho_k) \left( |{\mathbf{h}}^H_{{\rm dl},k} {\mathbf{v}}_k|^2 + \tilde{G}_k + \sigma^2_k \right)
  \geq \bar{Q}_k, \nonumber\\
\!\!\!& &\!\!\! \mathbf{H}^H_{{\rm dl},k} {\mathbf{v}}_k = 0, \|\mathbf{v}_k\|^2 \leq P_{max}, \nonumber\\
\!\!\!& &\!\!\!  0 < \rho_k < 1.
 \label{probs9}
\end{eqnarray}
We remark that for problem (\ref{probs9}), with the optimal ZF beamforming solution $\tilde{{\mathbf{v}}}^*_k,$ and PS solution $\tilde{\rho}^*_k,$ the SINR constraint and the harvested power constraint should both hold with equality by contradiction. Notice the following:
\begin{itemize}
\item[(i)] Supposing that both the SINR and harvested power constraints are not tight given  $\tilde{\rho}^*_k$ and $\tilde{{\mathbf{v}}}^*_k,$ this implies that there must be an $\alpha_k,$ $0 < \alpha_k < 1$ such that with the new solution ${\mathbf{v}}^*_k = \alpha_k \tilde{{\mathbf{v}}}^*_k,$ and $\rho^*_k = \tilde{\rho}^*_k,$ either the SINR or harvested power constraint is tight. This new solution gives rise to a reduction in the transmission power which contradicts the fact that $\tilde{{\mathbf{v}}}^*_k,$ and $\tilde{\rho}^*_k$ is optimal for problem (\ref{probs9}). Therefore, the case that both the SINR and harvested power constraints are not tight cannot be true \cite{joint_transmit}.
\item[(ii)] Also, the scenario where the SINR constraint is tight but the harvested energy constraint is not tight cannot be true as  $\tilde{\rho}^*_k$ can be increased slightly such that both SINR and harvested power constraints become not tight anymore.
\item[(iii)] Similarly, the conclusions drawn in \cite{joint_transmit} also verify that the case where the harvested power constraint is tight but the SINR constraint is not tight cannot be true.
\end{itemize}

To summarize, with the optimal solution using the ZF transmit beamforming constraint, for problem (\ref{probs9}), the SINR and harvested power constraints must both hold with equality. Accordingly, problem (\ref{probs9}) is equivalent to
\begin{eqnarray}
\!\!\!& &\!\!\! \min_{{{\mathbf{v}}_k,\rho_k}}  \|{\mathbf{v}}_k\|^2 \nonumber\\
\!\!\!& &\!\!\! \mathrm{ s.t.}\nonumber\\
\!\!\!& &\!\!\! \frac{\rho_k|{\mathbf{h}}^H_{{\rm dl},k} {\mathbf{v}}_k|^2}{\rho_k(\bar{G}_k + \sigma^2_k ) + \delta^2_k}
= \gamma^{\mathrm{{DL}}}_k, \nonumber\\
\!\!\!& &\!\!\! (1-\rho_k) \left( |{\mathbf{h}}^H_{{\rm dl},k} {\mathbf{v}}_k|^2 + \tilde{G}_k + \sigma^2_k \right)
  = \bar{Q}_k, \nonumber\\
\!\!\!& &\!\!\! \mathbf{H}^H_{{\rm dl},k} {\mathbf{v}}_k = 0, \|\mathbf{v}_k\|^2 \leq P_{max}, \nonumber\\
\!\!\!& &\!\!\!  0 < \rho_k < 1.
 \label{probs9a}
\end{eqnarray}
Notice from problem (\ref{probs9a}) that the first two equality constraints can be rearranged to give the following equation
\begin{equation}
\gamma^{\mathrm{{DL}}}_k \left(\bar{G}_k + \sigma^2_k + \frac{\delta^2_k}{\rho_k}\right) = \frac{\bar{Q}_k}{(1-\rho_k)}-\tilde{G}_k - \sigma^2_k.
\label{probs10a}
\end{equation}
After some manipulations, (\ref{probs10a}) can be written as
\begin{equation}
\alpha_k \rho^2_k - \beta_k \rho_k - C_k = 0,
\end{equation}
where
\begin{align}
\alpha_k &=  \gamma^{\mathrm{{DL}}}_k (\bar{G}_k + \sigma^2_k ) + \tilde{G}_k + \sigma^2_k,\\
\beta_k &= \gamma^{\mathrm{{DL}}}_k(\bar{G}_k + \sigma^2_k) + \tilde{G}_k + \sigma^2_k - \bar{Q}_k - \gamma^{\mathrm{{DL}}}_k\delta^2_k,\\
C_k &= -\gamma^{\mathrm{{DL}}}_k \delta^2_k.
\end{align}
The optimal solution satisfying $0 < \rho_k < 1$ is given by
\begin{equation}
\tilde{\rho}^*_k = {\frac{ + \beta_k \pm \sqrt{\beta^2_k + 4 \alpha_k C_k}}{2 \alpha_k}}.
\end{equation}		
Next, we define ${\mathbf{v}}_k = \sqrt{p_k} \tilde{{\mathbf{v}}}_k$ with $\| \tilde{{\mathbf{v}}}_k\| = 1, \forall k.$ Then problem (\ref{probs10a}) is equivalent to:
\begin{eqnarray}
\!\!\!& &\!\!\! \min_{{p_k, \tilde{{\mathbf{v}}}_k}} \hspace{1mm} p_k \nonumber\\
\!\!\!& &\!\!\! \mathrm{ s.t.}\nonumber\\
\!\!\!& &\!\!\! {p_k|{\mathbf{h}}^H_{{\rm dl},k} \tilde{{\mathbf{v}}}_k|^2}=\tau_k, \nonumber\\
\!\!\!& &\!\!\! \mathbf{H}^H_{{\rm dl},k} \tilde{{\mathbf{v}}}_k = 0, \nonumber\\
 \!\!\!& &\!\!\!   \| \tilde{{\mathbf{v}}}_k\| =1,
 \label{probs11}
\end{eqnarray}
where $\tau_k \triangleq \gamma^{\mathrm{{DL}}}_k \left(\bar{G}_k + \sigma^2_k + \frac{\delta^2_k}{\rho_k}\right).$ It is evident from the first constraint of (\ref{probs11}) that to achieve the minimum $p_k,$ the optimal $\tilde{{\mathbf{v}}}_k$ should be the optimal solution to the following problem:
\begin{eqnarray}
\!\!\!& &\!\!\! \max_{{\tilde{{\mathbf{v}}}_k}} \hspace{1mm} |{\mathbf{h}}^H_{{\rm dl},k} \tilde{{\mathbf{v}}}_k|^2\nonumber\\
\!\!\!& &\!\!\! \mathrm{ s.t.}\nonumber\\
\!\!\!& &\!\!\! \mathbf{H}^H_{{\rm dl},k} \tilde{{\mathbf{v}}}_k = 0, \nonumber\\
 \!\!\!& &\!\!\!   \| \tilde{{\mathbf{v}}}_k\| =1.
 \label{probs11a}
\end{eqnarray}
Result obtained in \cite{joint_transmit} shows that the unique (up to phase rotation) optimal solution to problem  (\ref{probs11a}) is given by
\begin{equation}
\tilde{{\mathbf{v}}}_k = \frac{\mathbf{U}_k \mathbf{U}^H_k {\mathbf{h}}_{{\rm dl},k}}{\|\mathbf{U}_k \mathbf{U}^H_k {\mathbf{h}}_{{\rm dl},k} \|},
\end{equation}
where  $\mathbf{U}_k$ denotes the orthogonal basis for the null space of $\mathbf{H}^H_k.$ Accordingly, the optimal power solution is given by\cite{joint_transmit}
\begin{equation}
p_k = \frac{\tau_k}{|{\mathbf{h}}^H_{{\rm dl},k} \tilde{{\mathbf{v}}}_k|^2} = \frac{\tau_k}{\|\mathbf{U}_k \mathbf{U}^H_k {\mathbf{h}}_{{\rm dl},k} \|^2}.
\end{equation}
Thus, it follows that $\tilde{{\mathbf{v}}}_k$ for problem
 (\ref{probs9a}) is given by
 \begin{equation}
 \tilde{{\mathbf{v}}}^*_k = \sqrt{\gamma^{\mathrm{{DL}}}_k \left(\bar{G}_k + \sigma^2_k + \frac{\delta^2_k}{\rho_k}\right)} \frac{\mathbf{U}_k \mathbf{U}^H_k {\mathbf{h}}_{{\rm dl},k}}{\|\mathbf{U}_k \mathbf{U}^H_k {\mathbf{h}}_{{\rm dl},k} \|^2}.
 \end{equation}

\end{document}